\documentclass[11pt]{article}

\usepackage{amsmath,amsfonts}
\usepackage{algorithm}
\usepackage{array}
\usepackage[caption=false,font=normalsize,labelfont=sf,textfont=sf]{subfig}
\usepackage{textcomp}
\usepackage{stfloats}
\usepackage{url}
\usepackage{verbatim}
\usepackage{graphicx}
\usepackage{cite}
\usepackage{amssymb}
\usepackage{amsthm}
\usepackage{xcolor}
\usepackage{multirow} 
\usepackage{algpseudocode}

\usepackage[a4paper, top=2cm, bottom=2cm, left=2cm, right=2cm]{geometry}

\usepackage{hyperref} 
\hypersetup{colorlinks=true,linkcolor=black,filecolor=black,urlcolor=black,citecolor=black}

\algnewcommand{\IIf}[1]{\State\algorithmicif\ #1\ \algorithmicthen}
\algnewcommand{\EndIIf}{\unskip\ \algorithmicend\ \algorithmicif}

\newcommand{\R}{\mathbb R}

\newcommand{\C}{\mathcal C}

\newcommand{\Lc}{\mathcal L}

\newcommand{\Hc}{\mathcal{H}}
\newcommand{\Sc}{\mathcal S}
\newcommand{\Uc}{\mathcal U}

\newcommand{\db}{\mathbf d}
\newcommand{\e}{\mathbf e}

\newcommand{\q}{\mathbf q}
\newcommand{\ub}{\mathbf u}
\newcommand{\vb}{\mathbf v}
\newcommand{\w}{\mathbf w}
\newcommand{\x}{\mathbf x}
\newcommand{\y}{\mathbf y}
\newcommand{\z}{\mathbf z}

\newcommand{\A}{\mathbf A}
\newcommand{\B}{\mathbf B}

\newcommand{\G}{\mathbf G}

\newcommand{\I}{\mathbf I}
\newcommand{\J}{\mathbf J}
\newcommand{\K}{\mathbf K}
\newcommand{\Lb}{\mathbf L}
\newcommand{\M}{\mathbf M}
\newcommand{\N}{\mathbf N}

\newcommand{\Q}{\mathbf Q}

\newcommand{\U}{\mathbf U}
\newcommand{\V}{\mathbf V}

\newcommand{\Z}{\mathbf Z}

\newcommand{\Lambdab}{\boldsymbol \Lambda}
\newcommand{\betab}{\boldsymbol \beta}
\newcommand{\gammab}{\boldsymbol \gamma}

\newcommand{\0}{\mathbf 0}
\newcommand{\1}{\mathbf 1}

\newcommand{\diag}{\mathrm{diag}}
\newcommand{\tr}{\mathrm{tr}}

\newcommand{\rank}{\mathrm{rank}}

\theoremstyle{definition}
\newtheorem{theorem}{Theorem}
\newtheorem{lemma}{Lemma}
\newtheorem{cor}{Corollary}

\title{Laplacian Network Optimization via Information Functions}
\author{Samuel Rosa and Radoslav Harman}
\date{\today}

\begin{document}

\maketitle

\begin{abstract}
Designing networks to optimize robustness and other performance metrics is a well-established problem with applications ranging from electrical engineering to communication networks. Many such performance measures rely on the Laplacian spectrum; notable examples include total effective resistance, the number of spanning trees, and algebraic connectivity. This paper advances the study of Laplacian-based network optimization by drawing on ideas from experimental design in statistics. We present a theoretical framework for analyzing performance measures by introducing the notion of information functions, which captures a set of their desirable properties. Then, we formulate a new parametric family of information functions, Kiefer's measures, which encompasses the three most common spectral objectives. We provide a regular reformulation of the Laplacian optimization problem, and we use this reformulation to compute directional derivatives of Kiefer's measures. The directional derivatives provide a unified treatment of quantities recurring in Laplacian optimization, such as gradients and subgradients, and we show that they are connected to Laplacian-based measures of node distance, which we call node dissimilarities. We apply the node dissimilarities to derive efficient rank-one update formulas for Kiefer's criteria, and to devise a new edge-exchange method for network optimization. These update formulas enable greedy and exchange algorithms with reduced asymptotic time complexity.
\end{abstract}

\textit{Keywords:}
Graph, network design, Laplacian spectrum, robustness,  reliability, effective resistance 

\section{Introduction}
\label{sIntro}

Laplacian spectral measures for assessing networks have been applied in a variety of fields, including electrical engineering, communication systems, robotics, and dynamical networks \cite{GhoshEA08, Tizghadam10, EchefuEA24, SiamiMotee}. Across these domains, such measures have meaningful interpretations in terms of connectivity, robustness, or network dynamics. Moreover, many Laplacian optimization problems can also be viewed as special cases of optimal experimental design problems studied in mathematical statistics, a connection noted decades ago \cite{Cheng81, GhoshEA08, KhosoussiEA14}. For conciseness, we refer to these fields as Network design (ND) and Experimental design (ED). 

Although the underlying problems in different areas of ND are often similar or even mathematically equivalent, analytical and algorithmic tools are frequently developed independently, with varying terminology. For instance, the effective resistance \cite{GhoshEA08} is also known as network criticality \cite{Tizghadam10}, the $\Hc_2$ norm \cite{YoungEA}, and the Kirchhoff index \cite{YangKlein}. Consequently, there are many ad-hoc results on Laplacian optimization across the fields of ND, with some fields lacking an overarching theoretical framework for the study of these problems. Furthermore, because ND is a relatively young field, some aspects of Laplacian optimization in ND are less developed than those in the ED literature. In this paper, we draw on concepts from ED to provide new contributions to ND, especially in extending the existing theoretical framework but also in providing a new optimization algorithm, as detailed in Section \ref{ssContrib}. To connect the various areas of ND that use Laplacian spectral measures, and to emphasize the equivalence of some studied problems, we first provide a short survey, summarizing the different terminologies and highlighting common results.

\subsection{Laplacian optimization in Network design}
\label{ssAreas}

The optimization of Laplacian spectra involves constructing a network so that a selected spectral measure is maximized (or minimized). In particular, we consider the problem of adding $N$ edges to a pre-existing (possibly empty) network.
The most widely used Laplacian eigenvalue-based measures fall into three classes: measures equivalent to (i) effective resistance, (ii) the number of spanning trees, and (iii) algebraic connectivity. Effective resistance corresponds to $A$-optimality, the number of spanning trees to $D$-optimality, and algebraic connectivity to $E$-optimality in the ED context. To denote these three classes clearly, we adopt the alphabetical notation from ED, which allows for a simple and unified treatment.

There are numerous motivations and interpretations for these performance measures. Below we give an overview. In each case, the measures are not necessarily identical, but the corresponding optimization problems are equivalent.
\begin{itemize}
	\item $A$-optimality: total effective resistance \cite{GhoshEA08}, expected average commute time \cite{GhoshEA08}, network criticality \cite{Tizghadam10}, 
	Kirchhoff index \cite{EllensEA, YangKlein}, $\Hc_2$ norm \cite{YoungEA, SiamiMotee, BhelaEA}, measure of network robustness \cite{EllensKooij}, measure of average centrality \cite{BozzoEA}.
	\item $D$-optimality: number of spanning trees, $t$-optimality, tree-connectivity \cite{KhosoussiEA20}, graph complexity \cite{KhosoussiEA20}, volume of feasible bus injections \cite{ThiamDeMarco}, measure of network robustness \cite{Baras}, uncertainty volume \cite{SiamiMoteeSchur}, steady-state entropy \cite{SiamiMoteeSchur}.
	\item $E$-optimality: algebraic connectivity \cite{GhoshEA08, WanEA}, measure of network robustness \cite{EllensKooij}, Hankel norm \cite{SiamiMotee18}.
\end{itemize}

We also provide an overview of the areas and corresponding optimization problems that lead to $A$-, $D$-, or $E$-optimality for network Laplacians, with the relevant measures indicated in parentheses:
\begin{itemize}
	\item Robustness and reliability of general networks (e.g. transportation, computer networks): maximizing robustness \cite{EllensKooij, ChanAkoglu, YamashitaEA, FreitasEA23, PizzutiSocievole, Baras} (A, D, E), maximizing reliability \cite{WeichenbergEA} (D).
	\item Consensus networks (Dynamical networks): optimal topology \cite{DaiMesbahi} (E), minimizing $\Hc_2$-norm \cite{YoungEA, SiamiMotee} (A), maximizing robustness \cite{SiamiMotee, SiamiMotee18} (A, D, E), minimizing uncertainty volume \cite{SiamiMoteeSchur, SiamiMotee18} (D), minimizing steady-state entropy \cite{SiamiMoteeSchur} (D). We note that the problem of \textit{linear} consensus networks leads to the Laplacian optimization problem considered in the present paper \cite{SiamiMotee, SiamiMotee18}; other types of consensus networks generally lead to more complex problems.
	\item Estimation on graphs, SLAM (Simultaneous localization and mapping): maximizing connectivity \cite{KhosoussiEA20, DohertyEA22, KavetiEA24, DohertyEA24} (D, E).
	\item General graph theory: minimizing the total effective resistance \cite{GhoshEA08} (A), maximizing algebraic connectivity \cite{GhoshEA08, GhoshBoyd, SomisettyEA24} (E), maximizing the number of spanning trees \cite{LiEA} (D).
	\item Electrical networks: minimizing the total effective resistance \cite{GhoshEA08} (A), maximizing the volume of feasible bus injections \cite{ThiamDeMarco} (D), minimizing the Kirchhoff index \cite{ZhouEA} (A). 
	\item Communication networks: minimizing network criticality \cite{Tizghadam10} (A), minimizing the average travel cost \cite{Tizghadam10} (A), maximizing robustness \cite{YoungEA} (A).
	\item Multi-vehicle or multi-robot patrolling: minimizing the average expected commute time \cite{AlamEA, EchefuEA24} (A).
	\item Markov chains: minimizing the average expected commute time \cite{GhoshEA08, ChandraEA89} (A).
	\item Network of pairwise comparisons: maximizing the information of a ranking \cite{OstingEA} (A, D, E).
	\item Air transportation: maximizing robustness \cite{YangEA19, WeiEA} (A, E).
	\item Graph neural networks: minimizing oversquashing \cite{BlackEA} (A).
	\item Friend recommendation in social networks: optimizing content spread \cite{Yu15} (E).
	\item Structure of isomers in chemistry: maximizing chemical reactivity \cite{YangKlein} (A).
\end{itemize}

Optimization of the Laplacian spectra is particularly well studied in robustness and reliability of general networks, dynamical networks, and SLAM. However, these three areas have largely developed independently, often arriving at similar results, while partial results also appear in other areas. Below, we present some notable examples, but a thorough analysis of findings in different areas of ND is beyond the scope of this paper. 

Computational methods for $A$-, $D$-, and $E$-optimal network problems typically involve sequential greedy algorithms or continuous convex relaxations solved by convex programming. Greedy methods have been developed independently in general networks \cite{PredariEA, ZhouEA}, consensus networks \cite{DekaEA}, SLAM \cite{KhosoussiEA20}, and air transportation \cite{WeiEA}. Some important results that are not widely known in all areas of network science include the submodularity of $D$-optimality, which enables guarantees on the quality of the greedy solution \cite{SiamiMotee18, KhosoussiEA19, KhosoussiEA20}, and methods for accelerating the greedy algorithm using randomization in multiple computations \cite{PredariEA, ZhouEA}. Mixed-integer programming methods have also been used \cite{BhelaEA, SomisettyEA24}, though they are generally limited to small network sizes. Convex relaxation approaches have been developed independently in general graph theory \cite{GhoshEA08}, robust networks \cite{Baras}, SLAM \cite{KhosoussiEA20, DohertyEA24}, dynamical networks \cite{SiamiMotee, SiamiMoteeSchur} and air transportation \cite{YangEA19}. These relaxed problems are typically solved using semidefinite programming \cite{GhoshBoyd} or the Frank-Wolfe method \cite{DohertyEA24}.

Analogous findings have also been obtained in ED, such as submodular guarantees for greedy methods \cite{Sagnol13}, convex relaxations \cite{puk} (which are known as approximate designs), semidefinite programming formulations \cite{VanderbergheBoyd}, and the Frank-Wolfe approach \cite{Wynn, Fedorov}.

To maintain focus, we exclude the convex relaxation approach from this paper and defer its more detailed examination to future research.

\subsection{Contributions and related work}
\label{ssContrib}

\paragraph{Overview of Laplacian spectra optimization} We provided a brief survey of optimization of Laplacian spectra for designing networks in Section \ref{ssAreas}, summarizing different terminologies and areas of ND. Previous works that contain notable surveys include
\cite{GhoshEA08}, which, however, predates many of the results mentioned above; and \cite{FreitasEA23}, which reviewed the robustness and reliability literature across multiple categories of networks, such as road networks. While many works acknowledge related research in other areas, the underlying equivalence of the optimization problems and the consequent similarity of results is often underemphasized.

\paragraph{Axiomatic study of network measures} We propose a set of desirable properties of performance measures, captured by the notion of information functions (analogously to ED; see \cite{puk}, Chapter 5). In dynamical networks, subsets of these properties were used to characterize ``systemic performance measures'' \cite{SiamiMotee, SiamiMotee18}. In other network domains, similar findings do not appear to have been established. Indeed, in the area of reliability and robustness, the “axiomatic study of desired properties in robustness measures” is identified as an important direction for future research \cite{FreitasEA23}.

\paragraph{New measures}
Inspired by ED (see Chapter 6 in \cite{puk}, Section 2.4.10 in \cite{LopezFidalgo2023}), we formulate Kiefer’s $\Phi_p$ measures, a natural parametric class of eigenvalue-based information functions that includes $A$-, $D$-, and $E$-optimality. This parametrization highlights both the differences and similarities among the three common measures and provides compromise measures (e.g., measures that lie between $A$- and $D$-optimality). A subclass of these criteria was given in dynamical networks as spectral zeta functions, which are equivalent to Hardy-Schatten system norms or $\Hc_p$ norms \cite{SiamiMotee, SiamiMotee18}. However, the spectral zeta functions do not include compromise measures between $A$- and $D$-optimality, nor do they include $D$-optimality itself. 

\paragraph{Regular reformulation and directional derivatives} We provide a regular reformulation of the optimal network problem, which allows for a convenient theoretical analysis. A similar technique has been used to address consensus networks \cite{YoungEA} and the subset of ED problems equivalent to network optimization \cite{HarmanFilova, SagnolHarman15}. Other regularization methods include using the reduced (or grounded) Laplacian or expressing the pseudoinverse of the Laplacian through the inverse of a non-singular matrix \cite{GhoshEA08}.

Using this reformulation, we compute directional derivatives of the Laplacian measures, which serve as an important tool for both theoretical analysis and algorithmic optimization of networks. They provide a common framework for quantities that are used \textit{ad hoc} in different areas, such as gradients, subgradients, and small network perturbations (e.g., \cite{Tizghadam10, TizghadamEA11, YangEA19} for $A$-optimality; \cite{Baras, ThiamDeMarco, ChanAkoglu} for $D$-optimality;  \cite{GhoshBoyd, Yu15, ChanAkoglu, DohertyEA24} for $E$-optimality). We show that the directional derivatives are closely tied to Laplacian-based measures of distance between nodes, which we call node dissimilarities. Moreover, we provide formulas for the directional derivatives of all $\Phi_p$ measures, which include $A$-, $D$-, and $E$-optimality as special cases. Directional derivatives were used to approximate changes in performance measures in dynamical networks \cite{SiamiMotee18}, where the directional derivatives for $A$-optimality, but not for the other $\Phi_p$ measures, are provided. 

\paragraph{Rank-one updates and efficient implementation} The removal or the addition of a single edge results in a rank-one update of the network's Laplacian. Using the node dissimilarities, we provide efficient update formulas for all $\Phi_p$ measures with integer values of $p$, including $A$- and $D$-optimality. For $A$- and $D$-optimality, such formulas are common \cite{SiamiMotee18, ZhouEA}, but they have not been known for other $\Phi_p$ measures. Consequently, we provide an efficient version of the greedy method, which has smaller asymptotic time complexity than a straightforward implementation.

\paragraph{Exchange algorithm} 
We propose the KL exchange method, which is an edge-exchange-based heuristic for network optimization that utilizes the node dissimilarities, inspired by exchange algorithms from ED (see \cite{AtkinsonEA07}, Section 12.6). Using the rank-one update formulas, we provide an efficient version of the method for $\Phi_p$ measures with improved time complexity. We then use examples to show that the KL exchange method can be used to improve upon the solution provided by the greedy algorithm, within a relatively short computation time. An edge-exchanging heuristic for optimizing $A$-optimality was given in \cite{HackettEA}, which, however, considers only a subset of possible exchanges utilized in our algorithm. Edge exchanges were also employed in the EdgeRewire method \cite{ChanAkoglu}, but it targets different, more restrictive optimization settings, including a fixed number of exchanges that preserve node degrees.

\section{Background}

\subsection{Matrix notation}

We denote matrices by uppercase bold letters and (column) vectors by lowercase bold letters. The symbols $\R^{m \times n}$, $\Sc^n$, and $\Sc^n_+$ denote the sets of all $m \times n$ matrices, $n \times n$ symmetric matrices, and $n \times n$ nonnegative definite matrices, respectively. For a matrix $\A$, $\C(\A)$ is its column space, $\A^+$ is its Moore-Penrose pseudoinverse, and if $\A \in \Sc^{n}$, then $\lambda_1(\A) \geq \ldots \geq \lambda_n(\A) $ are its eigenvalues. For $\A \in \Sc^{n}$, $\A^0$ is the projection matrix to $\C(\A)$. The superscript $\perp$ denotes the orthogonal complement. The Loewner ordering of symmetric matrices is denoted by $\succeq$ (i.e. $\A \succeq \B$ iff $\A - \B \in \Sc^n_+$). The diagonal matrix with elements of $\x$ on the diagonal is denoted by $\diag(\x)$. Symbols $\I_n$ and $\J_n$ denote the $n\times n$ identity matrix and the $n\times n$ matrix of all ones, respectively. The matrix or the vector of zeros is denoted by $\0$. The $n \times 1$ vector of ones is denoted by $\1_n$, and $\e_i$ is the vector with $i$th element equal to one and all other elements zero. 

\subsection{Graph theory}

Let $G=(V, E, \w)$ be a (simple undirected) weighted graph, where $V = \{1, \ldots, n\}$ is the set of nodes, $E \subseteq {V \choose 2}$ is the set of $m$ edges with weights $w_e\geq 0$, $e \in E$. To simplify the notation, we identify the edge weights with the nonnegative vector $\w \in \R^m$. The incidence matrix of $G$ is an $n \times m$ matrix $\Q$, which is obtained by arbitrarily assigning orientations to the edges, and then its elements are $q_{ie} = 1$ if edge $e$ leads to node $i$, $q_{ie} = -1$ if $e$ leads from $i$, and $q_{ie} = 0$ otherwise. The Laplacian $\Lb \in \R^{n \times n}$ is then given by $\Lb = \Q \diag(\w) \Q^T \in \R^{n \times n}$ and its elements satisfy 
\begin{equation*}
\ell_{i,j} =
\begin{cases}
	\sum_{e \sim i} w_e, &i = j, \\
	-w_{(i, j)}, & (i, j) \in E, \\
	0, &\text{otherwise},
\end{cases}
\end{equation*}
where $e \sim i$ denotes edges $e$ incident with $i$. The Laplacian $\Lb$ always has at least one eigenvalue equal to zero (i.e., $\rank(\Lb) \leq n-1$), and $\Lb$ attains the maximum rank of $n-1$ iff the graph $G$ is connected. 

\subsection{Optimal networks}

The Laplacian optimization problem for network design can be stated as follows. For a given set of nodes $V$, edge weights $\w$ on the complete graph $K_n$, and pre-existing edges $E_0 \subset {V \choose 2}$, choose $N$ edges $E$ among ${V \choose 2} \setminus E_0$ so that $\Phi(\Lb_E)$ is maximized. Here, $\Lb_{E}$ is the weighted Laplacian corresponding to $G = (V, E_0 \cup E, \w)$, and $\Phi$ is a pre-selected metric measuring the ``quality'' of the network. With a slight abuse of notation, the weights $\w$ in $G$ are obtained by restricting the original weights $\w$ to the subgraph $G$ of the complete graph $K_n = (V, {V \choose 2}, \w)$. Taking edge weights into consideration is a common approach \cite{SiamiMotee, KhosoussiEA20, DohertyEA24}; such weights model the strengths or the importance of the connections (e.g., the precision of the measurements in SLAM \cite{KhosoussiEA20}).

Formally,
\begin{align}
	\max_{E} \quad &\Phi(\Lb_E) \label{ePhiopt}\\
	s.t. \quad & E \subseteq {V \choose 2} \setminus E_0, \notag \\
	&\vert E \vert = N. \notag
\end{align}
A $\Phi$-optimal network is denoted as $E^*$.
Problems without pre-existing edges $E_0$ are achieved by simply setting $E_0 = \emptyset$, and those without edge weights by setting $w_e = 1$ for all edges $e$. For simplicity, we refer to the network (or graph) $G = (V, E_0 \cup E, \w)$ as the network (or graph) $E$.

Some of the most popular functions $\Phi$ are based on the nontrivial eigenvalues $\lambda_1 \geq \ldots  \geq \lambda_{n-1}$ of $\Lb$. Suppose that $\Lb$ has rank $n-1$ (i.e., the network is connected). Then the three most popular measures are $\Phi_A(\Lb) = (n-1)(\sum_{i=1}^{n-1} \lambda_i^{-1})^{-1}$, $\Phi_D(\Lb) = (\prod_{i=1}^{n-1} \lambda_i(\Lb))^{1/(n-1)}$ and $\Phi_E(\Lb) = \lambda_{n-1}$. These three measures have different names in different contexts (see Section \ref{ssAreas}); in this paper, we denote them as $A$-, $D$- and $E$-optimality, respectively. Besides their domain-specific interpretations, $A$- and $D$-optimality represent the harmonic mean and the geometric mean of the positive Laplacian eigenvalues, respectively. Although there exist other versions of these criteria (e.g., $\log(\prod_{i=1}^{n-1} \lambda_i)$ for $D$-optimality), we deliberately use these versions because they satisfy the conditions of information functions (Section \ref{sInfFun}).

\subsection{Optimal experimental designs}

The usual aim of optimal experimental design is to select $N$ trials to maximize the information obtained from the experiment. Formally, a set of $k$ design points $\x_1, \ldots, \x_k \in \R^n$ is given and the optimization problem is
\begin{align}
	\max_{\db} \quad &\Phi(\M(\db)) \label{eOptDes}\\
	s.t. \quad & \M(\db) = \M_0 + \sum_{e=1}^k d_e \x_e \x_e^T, \notag \\
	& \db \in \{0,1\}^k, \quad \sum_{e=1}^k d_e = N. \notag
\end{align}
Here, $\db$ is the design (also called ``exact'' design), $d_e$ represents the number of trials performed at design point $\x_e$, $\M(\db)$ is the information matrix capturing the amount of information conveyed by running the experiment according to $\db$, the matrix $\M_0$ represents some form of prior information, and $\x_e \x_e^T$ is an elementary information matrix corresponding to a single trial at $\x_e$. The value $\Phi(\M(\db))$ then quantifies the amount of information as a real number. A $\Phi$-optimal design, i.e., a solution to \eqref{eOptDes}, is denoted as $\db^*$.

The simplest statistical motivation for \eqref{eOptDes} is that this problem corresponds to finding the optimal design for the linear regression model $y = \x^T \betab + \varepsilon$, where independent homoscedastic observations $y$ are collected at those design points $\x = \x_e$ for which $d_e = 1$. In this setting, $\M(\db)$ is a constant multiple of the Fisher information matrix, and if it is non-singular, then $\M^{-1}(\db)$ is proportional to the covariance matrix of the least-squares estimator of $\betab$. Problem \eqref{eOptDes} therefore represents the minimization of some real-valued aspect of this covariance matrix, making the estimator as precise as possible in the sense defined by $\Phi$. For more information, see the monographs \cite{Fedorov, Pazman86, puk, AtkinsonEA07, LopezFidalgo2023}.

To compare \eqref{ePhiopt} and \eqref{eOptDes}, let $E_0^c = {V \choose 2} \setminus E_0$, $k = \lvert E_0^c \rvert$, and let $\Lb_0$ be the Laplacian of the initial network $(V, E_0, \w)$. For an edge $e = (i,j)$, let $\q_e = \e_i - \e_j$ and $\x_e = \sqrt{w_e} \q_e$. Note that we use this notation throughout the paper. The Laplacian $\Lb_E$ can then be expressed as 
\begin{equation*}
\Lb_E = \Lb_0 + \sum_{e \in E_0^c} d_e w_e \q_e \q_e^T
= \Lb_0 + \sum_{e \in E_0^c} d_e \x_e \x_e^T,
\end{equation*}
where $\db$ is the network design with $d_e = 1$ if $e \in E$, $d_e = 0$ otherwise. It follows that the Network design problem \eqref{ePhiopt} is a special case of \eqref{eOptDes}, as noted in Section \ref{sIntro}. In fact, the optimal network problem corresponds to the optimal block design problem (with blocks of size two), where the nodes represent the treatments and the edges represent the blocks \cite{Cheng81}. This problem is also closely related to ED for the Bradley-Terry model \cite{GrasshoffEA04, GrasshoffSchwabe}.

Although Laplacian optimization corresponds to an important class of problems in ED, it is not entirely typical of experimental design problems. Most notably, while the Laplacian matrix has rank at most $n-1$, optimal information matrices in ED are usually of full rank (with block designs being an exception). As a result, a form of regularization (Section \ref{ssRegular}) is needed to translate results from ED to the ND setting. Moreover, the dimensions of the information matrices in ED are generally much smaller than those of the Laplacian matrices encountered in ND.

Similarly to ND, the most popular measures in ED (``optimality criteria'') are $A$- and $D$-optimality, with $E$-optimality also being among the most prominent ones. From the ED perspective, the condition $d_e \in \{0,1\}$ corresponds to binary (replication-free) designs, which allow at most one trial at each design point. Allowing $d_e$ to take values in $\{0,1,\ldots,N\}$ instead enables multiple trials at each design point. The assumption of pre-existing edges $E_0$ that must be present in the final network corresponds to protected, required, or existing trials in the experimental design context. Optimal experimental designs for such situations are called optimally augmented designs (Section 19 in \cite{AtkinsonEA07}) or optimal designs with protected runs (Section 11.7 in \cite{puk}), and they are closely related to optimal Bayesian designs (Section 11 in \cite{puk}; \cite{ChalonerVerdinelli}). The convex relaxations in ED correspond to the so-called approximate experimental designs \cite{puk}, which are obtained by letting $d_e \in [0,1]$ and $\sum_{e} d_e = 1$ instead of $d_e \in \{0,1\}$ and $\sum_{e} d_e = N$.

\section{Information functions}\label{sInfFun}
We propose that reasonable Laplacian measures are information functions, which are defined by the following conditions. Let $\Lc^n$ be the set of all $n \times n$ potential Laplacians: $\Lc^n = \{ \Lb \in \Sc^n_+ \, \vert \, \Lb \1_n = \0 \}$, which is a convex cone in $\Sc^n_+$. Then, a function $\Phi: \Lc^n \to \R$ is an information function on $\Lc^n$ if and only if it is non-constant and satisfies the following properties: 
\begin{enumerate}
    \item Isotonicity: if $\Lb_1 \succeq \Lb_2$ for some $\Lb_1,\Lb_2 \in \Lc^n$, then $\Phi(\Lb_1) \geq \Phi(\Lb_2)$;
	\item Positive homogeneity: $\Phi(\alpha \Lb) = \alpha \Phi(\Lb)$ for all $\Lb \in \Lc^n$ and $\alpha > 0$;
	\item Concavity: $\Phi(\alpha \Lb_1 + (1-\alpha) \Lb_2) \geq \alpha \Phi(\Lb_1) + (1-\alpha) \Phi(\Lb_2)$ for all $\Lb_1,\Lb_2 \in \Lc^n$ and $\alpha \in [0,1]$;
	\item Upper semicontinuity: the level sets $\{\Lb \in \Lc^n \, \vert \, \Phi(\Lb) \geq \alpha\}$ are closed for all $\alpha \in \R$.
\end{enumerate}

Note that information functions are necessarily nonnegative: $\Phi(\Lb) \geq 0$ for all $\Lb \in \Lc^n$.
Isotonicity of an information function $\Phi$ directly implies that $\Phi$ does not decrease when any number of edges are added to the network. Positive homogeneity enables a meaningful comparison of networks through the ratios of their measure values. Indeed, for any two networks $E_1$ and $E_2$, and any positively homogeneous $\Phi$, the ratio $\Phi(\Lb_{E_1}) / \Phi(\Lb_{E_2})$ remains unchanged if we multiply the edge weights by a common constant $\alpha > 0$ (or if we replicate all edges $\alpha$ times when $\alpha$ is an integer). That is, positive homogeneity ensures that the network comparisons do not depend on whether the weights in the graph are normalized or not.

Concavity together with positive homogeneity guarantee that $\Phi$ is superadditive: $\Phi(\Lb_1 + \Lb_2) \geq \Phi(\Lb_1) + \Phi(\Lb_2)$. This captures the ``synergistic'' aspect of network performance: the quality of the union of two networks with disjoint edge sets is at least the sum of the individual network qualities. Upper semicontinuity is a more technical property, which (together with the other properties) guarantees that there exists a $\Phi$-optimal Laplacian $\Lb^*$ on any compact subset of $\Lc^n$. See Chapter 5 of \cite{puk} for a more detailed analysis of the information functions.

Functions $\Phi_A$, $\Phi_E$ and $\Phi_D$, as defined above, are information functions. The $D$-optimality measure is often defined in the logarithmic form $\log(\prod_{i=1}^{n-1}\lambda_i)$, which, however, is not positively homogeneous and thus not an information function. The logarithmic version of $D$-optimality is therefore not appropriate when edge weights are present and one wishes to compare networks via ratios of measures.

In this paper, we focus on information functions $\Phi$ that depend exclusively on the eigenvalues of $\Lb$, which clearly covers $A$-, $E$- and $D$-optimality. This requirement is equivalent to orthogonal invariance, i.e., $\Phi(\U \Lb \U^T) = \Phi(\Lb)$ for any $\Lb \in \Lc^n$ and any orthogonal $n \times n$ matrix $\U$ \cite{Harman04}. 

The ``convex systemic measures'' proposed for dynamical networks \cite{SiamiMotee} correspond exactly to the information functions (albeit in a minimizing rather than maximizing setting), except that our definition additionally requires semicontinuity and nonconstancy. Similarly, the ``orthogonally invariant systemic performance measures'' \cite{SiamiMotee18} correspond to orthogonally invariant (eigenvalue-based) information functions, but they do not impose the properties of positive homogeneity, upper semicontinuity and nonconstancy.

\section{Kiefer's measures}
\label{sKiefer}

We define a rich class of orthogonally invariant information functions, which covers $A$-, $D$- and $E$-optimality: Kiefer's $\Phi_p$ measures
\begin{equation*}
\Phi_p(\Lb) = \begin{cases}
	(\prod_{i=1}^{n-1} \lambda_i(\Lb))^{1/(n-1)}, &p=0; \\
	(\frac{1}{n-1} \sum_{i=1}^{n-1} \lambda_i^{-p}(\Lb))^{-1/p}, &p \in (0, \infty); \\
	\lambda_{n-1}(\Lb), &p = \infty
\end{cases}
\end{equation*}
for $\Lb \in \Lc^n$ of rank $n-1$. For $\Lb \in \Lc^n$ of lower rank, we define $\Phi_p(\Lb) = 0$.
$A$-, $D$- and $E$-optimality are obtained for $p=1$, $p=0$ and $p=\infty$, respectively. The class $\Phi_p$ highlights the connection between the three most popular measures, because the limit of $\Phi_p(\Lb)$ is $\Phi_0(\Lb)$ and $\Phi_{\infty}(\Lb)$ as $p \to 0_+$ and $p \to \infty$, respectively (see Sections 6.6 and 6.7 in \cite{puk}). Similarly, they also provide ``compromise'' measures between $A$-, $D$- and $E$-optimality. Equivalent measures were provided for dynamical networks \cite{SiamiMotee, SiamiMotee18}, but only for $p \geq 1$. 

For connected networks, it can be convenient to equivalently express Kiefer's criteria for $p \in (0, \infty)$ as
\begin{equation*}
\Phi_p(\Lb) = \left(\frac{\tr((\Lb^+)^{p})}{n-1}\right)^{-1/p}.
\end{equation*}
Another reformulation of the criteria, which is particularly useful for practical computations, is based on the identity $\Lb^+ = (\Lb + \J_n/n)^{-1} - \J_n/n$ for connected networks \cite{GhoshEA08}. Then, the matrix $\N := \Lb + \J_n/n$ has the same positive eigenvalues as $\Lb$, with the addition of one eigenvalue equal to one. Consequently, $\Phi_p(\Lb)$ for $\Lb$ of rank $n-1$ can be expressed in terms of $\N$:
\begin{equation}\label{ePhipB}
	\Phi_p(\Lb) = \begin{cases}
		\det(\N)^{1/(n-1)}, &p=0; \\
		(\frac{1}{n-1} (\tr(\N^{-p}) - 1 ))^{-1/p}, &p \in (0, \infty); \\
		\lambda_{n-1}(\N - \J_n/n), &p = \infty.
	\end{cases}
\end{equation}
Note that $\lambda_{n-1}(\N - \J_n/n)$ is simply the smallest eigenvalue of $\N$ after discarding one eigenvalue exactly equal to one.

\section{Regular reformulation}
\label{ssRegular}

The Laplacian matrix is always singular, which complicates the analysis of \eqref{ePhiopt}; we therefore propose a regular reformulation for eigenvalue-based criteria. We find that this reformulation is more convenient than using the matrix $\N$ above for the theoretical analysis of such criteria, especially for translating ED results to ND. Let $\K \in \R^{n \times (n-1)}$ be any fixed matrix whose columns are orthonormal and all orthogonal to $\1_n$. Then, instead of optimizing the eigenvalues of $\Lb$, one can equivalently optimize the eigenvalues of $\widetilde{\Lb}:=\K^T\Lb\K$, which is equal to $(\K^T\Lb^+\K)^{-1}$ for connected networks, as stated in the following theorem.

\begin{theorem}\label{tReg}
	Let $\Lb$ be a Laplacian of rank $n-1$, let $\K \in \R^{n \times (n-1)}$ be a matrix with orthonormal columns that are all orthogonal to $\1_n$, and $\widetilde{\Lb} = \K^T \Lb \K$. Then (a) $\widetilde{\Lb}$ is a positive definite matrix, (b) $\K^T\Lb^+\K$ is positive definite, and $\widetilde{\Lb} = (\K^T\Lb^+\K)^{-1}$, and (c) the positive eigenvalues of $\Lb$ and $\widetilde{\Lb}$ are the same, including multiplicities.
\end{theorem}

Proof of Theorem \ref{tReg} and all other proofs are deferred to Appendix \ref{apProofs}. From a statistical point of view, the matrix $\K^T\Lb\K$ corresponds to the information matrix of a reparametrized model with $\widetilde{\x}_e = \K^T \x_e$, and $(\K^T\Lb^+\K)^{-1}$ is the information matrix in the original model, but for the parameters of interest $\K^T\betab$ (cf. \cite{puk}, Section 3). By Theorem \ref{tReg}, we have $\Phi_p(\Lb) = \widetilde{\Phi}_p(\widetilde{\Lb})$ for any $\Lb \in \Lc^n$ of rank $n-1$, where
$\widetilde{\Phi}_p$ is the full-rank version of the $\Phi_p$-criterion:
\begin{equation*}
\widetilde{\Phi}_p(\widetilde{\Lb}) = \begin{cases}
	(\det(\widetilde{\Lb}))^{1/(n-1)}, &p=0; \\
	(\frac{1}{n-1} \tr(\widetilde{\Lb}^{-p}))^{-1/p}, &p \in (0, \infty); \\
	\widetilde{\lambda}_{n-1}(\widetilde{\Lb}), &p = \infty,
\end{cases}
\end{equation*}
where $\widetilde{\lambda}_{n-1}(\widetilde{\Lb})$ is the smallest eigenvalue of the non-singular $\widetilde{\Lb} \in \R^{(n-1) \times (n-1)}$. Consequently, $E$ is $\Phi_p$-optimal if and only if it is optimal for the regularized problem
\begin{align}
	\max_{E} \quad &\widetilde{\Phi}_p(\widetilde{\Lb}_E) \label{eOptReg}\\
	s.t. \quad & \widetilde{\Lb}_E = \K^T\Lb_0\K + \sum_{e \in E} \widetilde{\x}_e\widetilde{\x}_e^T, \notag \\
	\quad &  E \subseteq {V \choose 2} \setminus E_0; \quad \lvert E \rvert = N. \notag
\end{align} 

\section{Directional derivatives}
\label{ssDirDer}

The reparametrization from Section \ref{ssRegular} enables computation of the directional derivative of any spectral measure $\Phi$ at $\Lb_1$ in the direction of $\Lb_2$ (more properly, in the direction $\Lb_2-\Lb_1$):
\begin{equation*}
\partial\Phi(\Lb_1, \Lb_2) = \lim_{\alpha \to 0_+} \frac{\Phi(\Lb_1 + \alpha (\Lb_2 - \Lb_1)) - \Phi(\Lb_1)}{\alpha},
\end{equation*}
where $\Lb_1$ and $\Lb_2$ are Laplacians. Directional derivatives make it simple to apply the tools of differential calculus to spectral measures, despite their domain being restricted to the cone $\Lc^n$. An explicit expression for $\Phi_p$, $p \neq \infty$, is given in the following theorem, with results for $p = \infty$ presented later.

\begin{theorem}\label{tDirDer}
	Let $p \in [0, \infty)$, let $\Lb_1$ be the Laplacian of a connected graph and let $\Lb_2$ be a Laplacian. Then,
	\begin{equation}\label{ePartialDer}
		\partial\Phi_p(\Lb_1, \Lb_2) = \Phi_p(\Lb_1) 
		\left(\frac{\tr((\Lb_1^+)^{1+p}\Lb_2)}{\tr((\Lb_1^+)^{p})} - 1\right).
	\end{equation}
\end{theorem}

Note that for $p=0$, we have $\tr((\Lb_1^+)^{p}) = n-1$, because $(\Lb_1^+)^0$ is the projection matrix to $\1_n^\perp$. 
Of particular importance is the directional derivative in the direction of $\Lb_2=\Lb_1 + \x_e\x_e^T$. Intuitively, this directional derivative represents the rate of change in the function $\Phi_p$ when an ``infinitesimally small'' edge $e$ is added to the network, corresponding to updating $\Lb_1$ to $\Lb_1 + \delta \x_e\x_e^T$ for a small $\delta>0$.

\begin{cor}
	Let $p \in [0, \infty)$, let $\Lb_1$ be the Laplacian of a connected graph, and let $e \in {V \choose 2} \setminus E_0$. Then
	\begin{equation}\label{ePartialDerR1}
		\partial\Phi_p(\Lb_1, \Lb_1 + \x_e\x_e^T) = \frac{(\Phi_p(\Lb_1))^{1+p}}{n-1}\x_e^T(\Lb_1^+)^{1+p}\x_e.
	\end{equation}
\end{cor}

The quantity $\x_e^T(\Lb^+)^{1+p}\x_e$ plays an important role in the directional derivatives and also in multiple results below; we denote it as 
\begin{equation}\label{eVarFunp}
	v_{\Lb, p}(e) := \x_e^T(\Lb^+)^{1+p}\x_e = w_{(i, j)} (b_{ii} + b_{jj} - 2 b_{ij}),
\end{equation} 
where $e=(i,j)$, $b_{ij}$ are the entries of $(\Lb^+)^{1+p}$, and the second equality holds because $\x_e = \sqrt{w_{(i, j)}}(\e_i - \e_j)$. The value $v_{\Lb, p}(e)$ measures Laplacian-based distance between nodes $i$ and $j$, and we therefore call it node dissimilarity. In particular, $v_{\Lb, 0}(e) = \x_e^T \Lb^+ \x_e$ is the effective resistance (also called resistance distance \cite{Bapat}) between nodes $i$ and $j$ for $D$-optimality; and $v_{\Lb, 1}(e) = \x_e^T (\Lb^+)^2 \x_e$ is the squared biharmonic distance (cf. \cite{BlackEA24}) between $i$ and $j$ for $A$-optimality. In ED, $v_{\Lb, 0}(e)$ is called the variance (or sensitivity) function, because it is related to the variances of predicted responses. For brevity, we also use $v_{E,p}$ to denote $v_{\Lb_E, p}$. 

Because $E$-optimality is generally non-differentiable ($\Phi_\infty$ is non-differentiable at $\Lb$ if $\lambda_{n-1}(\Lb)$ has multiplicity greater than one), it requires a separate analysis.

\begin{theorem}\label{tDirDerE}
	Let $\Lb_1$ be the Laplacian of a connected graph and let $\Lb_2$ be a Laplacian. Let $s$ be the multiplicity of $\lambda_{n-1}(\Lb_1)$, let $\ub_1, \ldots, \ub_s$ be a set of orthonormal eigenvectors of $\Lb_1$ corresponding to $\lambda_{n-1}(\Lb_1)$, and let $\U = [\ub_1, \ldots, \ub_s] \in \R^{n \times s}$.
	Then,
	\begin{equation}\label{ePartialEgen}
		\partial\Phi_{\infty}(\Lb_1, \Lb_2) = \lambda_{s}(\U^T\Lb_2\U) - \lambda_{n-1}(\Lb_1).
	\end{equation}
	In particular, if $\lambda_{n-1}(\Lb_1)$ is simple, and $\ub = (u_1, \ldots, u_n)^T$ is the corresponding normalized eigenvector, then
	\begin{equation}
		\partial\Phi_{\infty}(\Lb_1, \Lb_2) = \ub^T \Lb_2 \ub - \lambda_{n-1}(\Lb_1). \label{ePartialE}
	\end{equation}
	Moreover,
	\begin{equation}\label{ePartialDerqE}
		\partial\Phi_{\infty}(\Lb_1, \Lb_1 + \x_e\x_e^T) = 
		\begin{cases}
			(\x_e^T\ub)^2, & \text{if $\lambda_{n-1}(\Lb_1)$ is simple}; \\
			0, & \text{otherwise}
		\end{cases}
	\end{equation}
	for any edge $e \in {V \choose 2} \setminus E_0$.
\end{theorem}

As demonstrated in Theorem \ref{tDirDerE}, the behavior of $E$-optimality at $\Lb_1$ changes depending on the multiplicity of $\lambda_{n-1}(\Lb_1)$.
Let $\ub$ be a Fiedler vector of $\Lb_1$ (i.e., one of the normalized eigenvectors corresponding to $\lambda_{n-1}(\Lb_1)$). Then, $n^{-1}\J_n + \ub\ub^T$ is one of the supergradients used in the construction of $\partial\Phi_{\infty}(\Lb_1, \Lb_2)$, whether $\lambda_{n-1}(\Lb_1)$ is simple or not (see the proof of Theorem \ref{tDirDerE}). This supergradient is consequently sometimes used as a ``representative'' supergradient in ND algorithms \cite{DohertyEA22}, \cite{DohertyEA24}. The corresponding value $b:=\ub^T\Lb_2\ub - \lambda_{n-1}(\Lb_1)$ equals $\partial\Phi_{\infty}(\Lb_1, \Lb_2)$ when $\lambda_{n-1}(\Lb_1)$ is simple, and $b$ is an upper bound on this directional derivative when $\lambda_{n-1}(\Lb_1)$ has a greater multiplicity. For $\Lb_2=\Lb_1 + \x_e\x_e^T$, $b$ becomes $(\x_e^T\ub)^2$. Consequently, we define the node dissimilarity using this representative value for directional derivatives:
\begin{equation}
	v_{\Lb, \infty}(e) = (\x_e^T\ub)^2 = w_{(i, j)} (u_i - u_j)^2,
\end{equation}
where $\ub$ is a fixed normalized eigenvector corresponding to $\lambda_{n-1}(\Lb)$, and $e = (i, j)$. As with $D$- and $A$-optimality, $v_{\Lb, \infty}(e)$ has a nice graph interpretation: it is the square of the Fiedler distance of nodes $i$ and $j$ \cite{Oliveira}.

As shown in \eqref{ePartialDerqE}, when $\lambda_{n-1}(\Lb_1)$ is not simple, $\partial\Phi_{\infty}(\Lb_1, \Lb_1 + \x_e\x_e^T) = 0$, which is not particularly useful (intuitively, a rank-one update is not enough to increase multiple equal eigenvalues), and $v_{\Lb, \infty}(e)$ is generally not equal to this directional derivative. This suggests that the use of $v_{\Lb, \infty}(e)$ (or of the corresponding supergradient, as in \cite{GhoshBoyd}, \cite{DohertyEA22}, \cite{DohertyEA24}) may not be ideal for constructing algorithms for $E$-optimal networks for problems in which $\lambda_{n-1}(\Lb)$ can be expected to achieve greater multiplicity during the run of the algorithm. Despite that, we still employ this approach in the $E$-optimality algorithms in Section \ref{sAlgs}, and we leave other methods to future research.

\section{Algorithms}

Suppose that the network $E$ is connected. As noted in Section \ref{sKiefer}, $\Lb^+_E = (\Lb_E + \J_n/n)^{-1} - \J_n/n$, and we base all computations on $\N_E := \Lb_E + \J_n/n$. The expression for $\Phi_p(\Lb_E)$ is given in \eqref{ePhipB}; $v_{E, p}(e) = \x_e^T\N_E^{-1-p}\x_e$ and $v_{E, \infty}(e) = (\x_e^T\ub)^2$, where $\ub$ is a normalized eigenvector corresponding to $\lambda_{n-1}(\N_E - \J_n/n)$ (which is the smallest eigenvalue of $\N_E$ after discarding one eigenvalue exactly equal to one).

\subsection{Rank-one updates for $\Phi_p$}\label{ssRankOne}

Network optimization methods often modify the network by adding or removing edges, or by exchanging pairs of edges \cite{PredariEA, ChanAkoglu, HackettEA}. Such changes result in rank-one updates of the Laplacian, or in pairs of rank-one updates. We show that for Kiefer's measures with integer values of $p$, all relevant quantities can then also be effectively updated; thus significantly reducing the computation time.

The addition or the removal of an edge $e$ results in updating $\N_E$ to $\N_E \pm \x_e\x_e^T$. The update of $\N_E^{-1}$ can be performed using the Sherman-Morrison formula:
\begin{equation}\label{eRankOneUpdate}
	(\N_E \pm \x_e\x_e^T)^{-1} 
	= \N_E^{-1} \mp \frac{\N_E^{-1}\x_e\x_e^T\N_E^{-1}}{1 \pm v_{E, 0}(e)}.
\end{equation}
In the case of $\N_E-\x_e\x_e^T$, \eqref{eRankOneUpdate} can only be used if $\N_E-\x_e\x_e^T$ is non-singular, i.e., if $v_{E, 0}(e) \neq 1$ (similarly for \eqref{eAoptUpdate} below). The measures of $A$- and $D$-optimality can also be easily updated using \eqref{eRankOneUpdate} and the Matrix determinant lemma (see \cite{Seber}, 15.10(d)), respectively:
\begin{equation}\label{eAoptUpdate}
	\tr((\N_E \pm\x_e\x_e^T)^{-1}) 
	= \tr(\N_E^{-1}) \mp \frac{v_{E, 1}(e)}{1 \pm v_{E, 0}(e)},
\end{equation}
\begin{equation}\label{eDoptUpdate}
	\det(\N_E \pm\x_e\x_e^T) = (1 \pm v_{E, 0}(e))\det(\N_E).
\end{equation}
That is, the rank-one updates to $\N_E^{-1}$ can easily be computed in $O(n^2)$ instead of $O(n^3)$ operations, which would result from a naive implementation. The updates to $\Phi_0(\Lb_E)$ and $\Phi_1(\Lb_E)$ can be performed in just $O(1)$ operations, because the calculation of each node dissimilarity $v_{E, j}(e)$ requires only $O(1)$ operations (see \eqref{eVarFunp}), provided that $\N_E^{-j}$ has already been computed.

For general $\Phi_p$ measures $(p \in (0,\infty))$, the formula for $\Phi_p(\Lb_E)$ involves $\N_E^{-p}$, which requires $O(n^3)$ operations. For integer $p$, however, the rank-one updates to $\Phi_p$ can be performed using the following theorem.
\begin{theorem}\label{tPowerTrace}
	Let $E$ be a connected network and let $e$ be an edge. Then,
	\begin{equation}\label{ePowerTrace}
		\tr((\N_E \pm \x_e\x_e^T)^{-p}) = \tr(\N_E^{-p}) + \sum_{k=1}^p \frac{(\mp 1)^k}{(1 \pm v_{E, 0}(e))^k} \sum_{\substack{i_0, \ldots, i_k \geq 0 \\ i_0 + \ldots + i_k = p-k}} v_{E, i_0+i_k+1}(e) \left( \prod_{j=1}^{k-1} v_{E, i_j + 1}(e)\right) 
	\end{equation}
	for any positive integer $p$. The last product in \eqref{ePowerTrace} is understood to be one for $k=1$; and for $\tr((\N_E - \x_e\x_e^T)^{-p})$, the formula is only valid if $v_{E,0}(e) \neq 1$ (i.e., if $E \setminus \{e\}$ is still connected).
\end{theorem}

When $\N_E^{-1}, \ldots, \N_E^{-p-1}$ have already been computed, the update of $\Phi_p(\Lb_E)$ using Theorem \ref{tPowerTrace} takes only $O(1)$ operations, as for $D$- and $A$-optimality. Efficient updates of $\N_E^{-p}$ can be done by employing the following theorem. For brevity, in this section, we understand the zeroth power of any $\A \in \R^{n \times n}$ as $\A^0 = \I$, in contrast to the rest of the paper.

\begin{theorem}\label{tPowerInv}
	Let $E$ be a connected network, let $e$ be an edge, and let $p$ be a positive integer. Then,
	\begin{equation}\label{ePowerInv}
		(\N_E \pm \x_e\x_e^T)^{-p} = \N_E^{-p} \mp \sum_{j=0}^{p-1} \left(\N_E^{-1} \mp \frac{\N_E^{-1}\x_e\x_e^T\N_E^{-1}}{1 \pm v_{E, 0}(e)} \right)^j \frac{\N_E^{-1}\x_e\x_e^T\N_E^{-(p-j)}}{1 \pm v_{E, 0}(e)}.
	\end{equation}
	For $(\N_E - \x_e\x_e^T)^{-p}$, \eqref{ePowerInv} is only valid if $v_{E,0}(e) \neq 1$ (i.e., if $E \setminus \{e\}$ is still connected).
\end{theorem}

The application of Theorem \ref{tPowerInv} requires maintaining the list of matrix powers $\N_E^{-1}, \N_E^{-2}, \ldots, \N_E^{-p}$. Then, the rank-one update for each of these matrix powers can be performed in $O(n^2)$ operations via \eqref{ePowerInv}, by using the matrix powers from the previous iteration and the already computed updated matrix powers; see function  \textsc{EffUpd} in Algorithm \ref{aUpdate}. To employ both Theorems \ref{tPowerTrace} and \ref{tPowerInv}, \textsc{EffUpd} can be used to maintain $\N_E^{-1}, \N_E^{-2}, \ldots, \N_E^{-p-1}$. We note that the update formulas are simple enough to allow a vectorized implementation (i.e., employing matrix and vector operations instead of loops), which can further speed up the computations.

\begin{algorithm}[H]
\caption{Efficient updates of integer powers of $\N_E^{-1}$}\label{aUpdate}
\begin{algorithmic}
\State \textbf{Input:} $p$, $\N_E^{-1}, \ldots, \N_E^{-p}$, $e$, $c \in \{-1,+1\}$
\State \textbf{Output:} $(\N_E + c\,\x_e\x_e^T)^{-1}, \ldots, (\N_E + c\,\x_e\x_e^T)^{-p}$
\Function{EffUpd}{$p$, $\N_E^{-1}, \ldots, \N_E^{-p}$, $e$, $c$}
    \State $\ub \gets \N_E^{-1}\x_e$
    \State $\vb \gets \ub / (1 + c\,v_{E,0}(e))$
    \State $\y_i \gets \N_E^{-1}\vb \quad (i = 1,\ldots,p-1)$
    \State $\B_1 \gets \N_E^{-1} - c\,\ub\vb^T$
    \State $\z_{1} \gets \B_1 \ub$
    \State $\B_2 \gets \N_E^{-2} - c\,\ub\y_{1}^T - c\,\z_{1}\vb^T$
    \For{$r = 3$ to $p$}
        \State $\z_{r-1} \gets \B_{r-1}\,\ub$
        \State $\B_r \gets \N_E^{-r}
                 - c\,\ub \y_{r-1}^T
                 - c\,\z_{r-1}\vb^T
                 - c \displaystyle\sum_{j=1}^{r-2} \z_j\,\y_{\,r-j-1}^T$
    \EndFor
    \State \Return $(\B_1, \ldots, \B_p)$
\EndFunction
\end{algorithmic}
\end{algorithm}

The above formulas show that the rank-one updates corresponding to the addition or removal of edge $e$ are closely tied to the node dissimilarities $v_{E, p}(e)$ for integer $p$. For the case of $E$-optimality, there is no explicit update formula, but $v_{E, \infty}$ approximates the change in $\Phi_\infty$ if $\lambda_{n-1}(\Lb_E)$ is simple: $\Phi_\infty(\Lb_E) = \lambda_{n-1}(\Lb_E) = \ub^T \Lb_E \ub$ and
\begin{equation*}
\Phi_\infty(\Lb_E \pm \x_e\x_e^T) 
\approx \ub^T (\Lb_E \pm \x_e\x_e^T) \ub 
= \lambda_{n-1}(\Lb_E) \pm (\x_e^T\ub)^2 
= \Phi_\infty(\Lb_E) \pm v_{E, \infty}(e),
\end{equation*}
where $\ub$ is the Fiedler vector. A similar approximation holds for other $\Phi_p$ criteria, as $\Phi_p(\Lb_E \pm \x_e\x_e^T) - \Phi_p(\Lb_E) \approx \pm \partial \Phi_p(\Lb_E, \Lb_E \pm \x_e\x_e^T)$. Since $\partial\Phi_p(\Lb_E, \Lb_E \pm \x_e\x_e^T)$ is proportional to $v_{E, p}(e)$, edges with large (or small) values of $v_{\Lb, p}(e)$ are likely to yield a large increase (or small decrease) in the objective when added (or removed). This observation is exploited in our proposed edge-exchange algorithm.

\subsection{Efficient greedy algorithm}
\label{sGreedy}

The greedy algorithms build networks by successively adding edges that maximize the increase in the value of $\Phi$. That is, starting with a connected graph $G=(V, E \cup E_0, \w)$ with $E=\emptyset$, the edge that maximizes $\Phi(\Lb_{E \cup \{e\}})$ over all $e \not\in (E \cup E_0)$ is added to $E$ in each iteration, until $\vert E \vert = N$. In this section, we therefore assume that $E_0$ is connected. 

The time complexity of a single iteration of a naive implementation of the greedy method for $\Phi_p$-optimality involves computing $\Phi(\Lb_E + \x_e\x_e^T)$ for each of the $k$ candidate edges. Because the straightforward computation of $\Phi(\Lb_E + \x_e\x_e^T)$ takes $O(n^3)$ operations, the time complexity of each iteration is $O(kn^3)$, and the overall time complexity of such a method is $O(Nkn^3) = O(Nn^5)$ because $k$ is at most $O(n^2)$.

We formulate a fast version of the greedy algorithm for Kiefer's measures with integer $p$, see Algorithm \ref{aGreedyeff}. The evaluation of $\Phi_p(\Lb_E + \x_e\x_e^T)$ is performed via \eqref{eAoptUpdate} for $A$-optimality, \eqref{eDoptUpdate} for $D$-optimality, and \eqref{ePowerTrace} for other integer $p$, which requires computing $v_{E, 0}(e), \ldots, v_{E, p}(e)$. Because the node dissimilarities $v_{E, j}(e)$ require $\N_E^{-j}$, matrices $\N_E^{-1}, \ldots, \N_E^{-p-1}$ are updated at the end of each iteration, which is done using \eqref{eRankOneUpdate} for $D$-optimality and \textsc{EffUpd} for other integer values of $p$, in $O(n^2)$ operations. The calculation of one node dissimilarity $v_{E, j}(e)$ has then complexity $O(1)$, and all such values for candidate edges $e$ are computed in $O(k)$ operations. The updates to $\Phi_p$ are then also performed in $O(k)$ operations. The time complexity of each iteration is thus reduced from the naive $O(k n^3)$ to $O(k + n^2)$. The initial computation of $\N_E, \N_E^{-1}, \ldots, \N_E^{-p-1}$ requires $O(n^3)$ operations, and the total time complexity of the efficient greedy method is therefore $O(n^3 + Nk + Nn^2)$, which can be simplified to $O(n^3 + Nn^2)$ because $k$ is at most $O(n^2)$; a reduction from the $O(Nkn^3) = O(Nn^5)$ operations for a naive version.

\begin{algorithm}[H]
\caption{Efficient greedy algorithm}\label{aGreedyeff}
\begin{algorithmic}
    \State \textbf{Input:} connected $(V, E_0, \w)$; integer $p$
    \State \textbf{Output:} approximate solution $E$ to \eqref{ePhiopt} with $\Phi=\Phi_p$
    \State $E \gets \emptyset$
    \State Compute $\N_E, \N_E^{-1},\ldots,\N_E^{-p-1}$
    \For{$i = 1$ to $N$}
        \State Compute $v_{E,j}(e)$ for all $e \notin (E \cup E_0)$ and $j = 0,\ldots,p$
        \State Compute $\Phi_p(\N_E + \x_e \x_e^T)$
               using \eqref{eAoptUpdate}, \eqref{eDoptUpdate}, or \eqref{ePowerTrace}
               for all $e \notin (E \cup E_0)$
        \State $e^* \gets \displaystyle\arg\max_{e \notin (E \cup E_0)}
               \Phi_p(\N_E + \x_e \x_e^T)$
        \State $E \gets E \cup \{e^*\}$
        \State $\N_E \gets \N_E + \x_{e^*}\x_{e^*}^T$;
        Compute $\Phi_p(\N_E)$
        \State $(\N_E^{-1},\ldots,\N_E^{-p-1})
               \gets \textsc{EffUpd}(p{+}1, \N_E^{-1},\ldots,\N_E^{-p-1}, e^*, +1)$
    \EndFor
\end{algorithmic}
\end{algorithm}

\subsection{KL exchange algorithm}
\label{sAlgs}

To improve a network that already has the required number of edges, for example one obtained by a greedy algorithm, we propose an edge-exchange heuristic called the KL exchange algorithm. We formulate the method for the $\Phi_p$ measures ($p \in [0, \infty]$), but it can be extended to other functions, provided that the appropriate directional derivatives are known. In each iteration, the algorithm replaces an existing edge with one not currently present in the network. To avoid considering all $O(n^2)$ pairs of edges, the algorithm first finds $K$ edges from $E$, whose removal should result in only a small decrease in $\Phi$, and $L$ edges from $(E \cup E_0)^c$, whose addition should result in a large increase in $\Phi$. Because $\Phi_p(\Lb_E \pm \x_e\x_e^T)$ is approximately proportional to the node dissimilarity $v_{E, p}(e)$ (see Section \ref{ssRankOne}), we consider the edges with the $K$ smallest values of $v_{E, p}(e)$ for removal, and the edges with the $L$ largest values for the addition to the network. The edge pair among the $KL$ candidates that maximizes the increase in $\Phi_p(\Lb_E)$ is then selected; see Algorithm \ref{aKL}. This process is repeated until no increase in $\Phi_p$ can be achieved. Unlike the greedy methods, the KL exchange algorithm can also be used when $(V, E_0, \w)$ is not connected, albeit the input network $(V, E_0 \cup E_1, \w)$, i.e., the initial feasible solution for the algorithm, must be connected. Provided that $(V, E_0, \w)$ is not already connected, this can be achieved, for instance, by utilizing a spanning tree algorithm in the construction of the initial $E_1$.

\begin{algorithm}[H]
\caption{KL exchange algorithm}\label{aKL}
\begin{algorithmic}
    \State \textbf{Input:} initial $E_1$: $(V, E_0 \cup E_1, \w)$ connected; $\delta > 0$; $K$, $L$, $p$
    \State \textbf{Output:} approximate solution $E$ to \eqref{ePhiopt} with $\Phi=\Phi_p$
    \State $E \gets E_1$; $\Lb \gets \Lb_E$
    \State $\textit{cont} \gets \text{TRUE}$
    \While{\textit{cont}}
        \State Compute $v_{E,p}(e)$ for all edges $e$
        \State $I_K \gets$ the $K$ edges in $E$ with smallest $v_{E,p}(e)$
        \State $I_L \gets$ the $L$ edges not in $(E \cup E_0)$ with largest $v_{E,p}(e)$
        \State $(e^*, f^*) \gets \displaystyle\arg\max_{e \in I_K, f \in I_L}\,
            \Phi_p\left(\Lb - \x_e \x_e^T + \x_f \x_f^T\right)$
        \State $\Lb_{\text{temp}} \gets \Lb - \x_{e^*} \x_{e^*}^T + \x_{f^*} \x_{f^*}^T$
        \If{$\Phi_p(\Lb_{\text{temp}}) > (1+\delta)\Phi_p(\Lb)$}
            \State $E \gets (E \cup \{f^*\}) \setminus \{e^*\}$
            \State $\Lb \gets \Lb_{\text{temp}}$
        \Else
            \State $\textit{cont} \gets \text{FALSE}$
        \EndIf
    \EndWhile
\end{algorithmic}
\end{algorithm}

The exchange algorithm in \cite{HackettEA} is analogous to the KL exchange algorithm with $K=1$ (i.e., only the single edge yielding the smallest decrease in $\Phi$ is considered for deletion in each iteration). In contrast, to consider all feasible pairs of edges in each iteration, one can simply set $K=N$ and $L=k-N$; the resulting method is called the Fedorov exchange algorithm in ED \cite{Fedorov}. There also exist numerous modifications of the KL exchange algorithm in ED (e.g., \cite{AtkinsonEA07}, Section 12).

\subsection{Efficient KL exchange algorithm}

To speed up the KL exchange algorithm, we terminate the maximization of $\Phi_p(\Lb_E - \x_e\x_e^T + \x_f\x_f^T)$ prematurely: the exchange is performed once \textit{any} pair of $e \in I_K, f \in I_L$ that results in a non-negligible increase in $\Phi_p$ is found. We also note that because the KL exchange method is a heuristic, a multistart approach (i.e., repeatedly starting the method from different randomly generated $E_1$) can be used to achieve a better network. For integer values of $p$, we provide an efficient version of the KL exchange method (see Algorithm \ref{aKLeff}), which makes use of the update formulas detailed in Section \ref{ssRankOne}. For brevity, in these algorithms, $\Phi_p(\N)$ stands for $\Phi_p(\Lb)$ computed using \eqref{ePhipB}.

In the KL algorithm, $\N_E^{-1}, \ldots, \N_E^{-p-1}$ are initially computed using $O(n^3)$ operations. In each iteration, the values $v_{E, j}(e)$ ($j=0,\ldots, p$) are computed in $O(N + k)$. For each potential addition of edge $f \in I_L$, $\N_E^{-1}, \ldots, \N_E^{-p-1}$ are updated in $O(n^2)$ using \eqref{eRankOneUpdate} or \eqref{ePowerInv}, resulting in $O(Ln^2)$ operations, which allows for updating the $\Phi_p$ measure after the subsequent removals of edges $e \in I_K$. For each $f \in I_L$, the value of $\Phi_p(\Lb - \x_e\x_e^T + \x_f\x_f^T)$ for each $e \in I_K$ is then computed via \eqref{eAoptUpdate}, \eqref{eDoptUpdate}, or \eqref{ePowerTrace}, taking $O(KL)$ operations in total. Once the best exchange is selected, $\N_E^{-1}, \ldots, \N_E^{-p-1}$ are updated in $O(n^2)$. The overall complexity of each iteration is therefore reduced from $O(KLn^3)$ to $O(N + k +Ln^2 + KL + n^2)$, which can be simplified to $O(Ln^2)$ because $k$, $N$ and $K$ are at most $O(n^2)$.

\begin{algorithm}[t]
\caption{Efficient KL exchange algorithm}\label{aKLeff}
\begin{algorithmic}
    \State \textbf{Input:} initial $E_1$: $(V, E_0 \cup E_1, \w)$ is connected; $\delta>0$; $K$, $L$; integer $p$
    \State \textbf{Output:} approximate solution $E$ to \eqref{ePhiopt} with $\Phi=\Phi_p$
    \State $E \gets E_1$;
     $\N \gets \N_E$;
     $\Phi^* \gets \Phi_p(\N)$
    \State $\textit{cont} \gets \text{TRUE}$
    \State $\M_j \gets \N^{-j}$ for $j = 1,\ldots,p+1$
    \While{\textit{cont}}
        \State $\textit{cont} \gets \text{FALSE}$
        \State Compute $v_{E,j}(e)$ for all edges $e$ and for $j = 0,\ldots,p$
        \State $I_K \gets$ the $K$ edges in $E$ with smallest $v_{E,p}(e)$
        \State $I_L \gets$ the $L$ edges not in $(E \cup E_0)$ with largest $v_{E,p}(e)$
        \ForAll{$f \in I_L$}
            \State $\N^{(f)} \gets \N + \x_f\x_f^T$
            \State $(\M_1^{(f)},\ldots,\M_{p+1}^{(f)})
                   \gets \textsc{EffUpd}(p{+}1, \M_1,\ldots,\M_{p+1}, f, +1)$
            \ForAll{$e \in I_K$}
                \State Compute $\Phi^{(e,f)} = \Phi_p(\N^{(f)} - \x_e\x_e^T)$
                       using $\N^{(f)}$, $\M_j^{(f)}$ and
                       \eqref{eAoptUpdate}, \eqref{eDoptUpdate} or \eqref{ePowerTrace}
                \If{$\Phi^{(e,f)} > (1+\delta)\Phi^*$}
                    \State $E \gets (E \cup \{f\}) \setminus \{e\}$
                    \State $\N \gets \N^{(f)} - \x_e\x_e^T$
                    \State $\Phi^* \gets \Phi^{(e,f)}$
                    \State $(\M_1,\ldots,\M_{p+1})
                           \gets \textsc{EffUpd}(p{+}1,
                           \M_1^{(f)},\ldots,\M_{p+1}^{(f)},
                           e, -1)$
                           \State $\textit{cont} \gets \text{TRUE}$
                           \State \textbf{break}
                \EndIf
            \EndFor
            \IIf{\textit{cont}} \textbf{break} \EndIIf
        \EndFor
    \EndWhile
\end{algorithmic}
\end{algorithm}

\section{Examples}
\label{sEx}

All computations in this section were performed in the statistical software R on a 64-bit Windows 11 system with an Intel Core i5-13400F processor.

\subsection{Small network}
\label{ssSmall}

We first demonstrate our results on a small example. We selected $n=10$ nodes and independently generated edge weights $\w$ from the uniform distribution on $(0,1)$. The initial network $(V, E_0, \w)$ was a random tree, and the considered objective was to add $N=5$ edges to this network. We applied the greedy algorithm to construct efficient networks with respect to $D$-, $A$-, $\Phi_{3}$- and $E$-optimality, and then used the KL exchange algorithm (Algorithm \ref{aKLeff}) with $K=N$, $L=36-N$ (i.e., effectively the Fedorov exchange algorithm) to potentially improve the greedy solution. That is, the output of the greedy algorithm was used as the initial solution for Algorithm \ref{aKLeff}. The networks obtained by the exchange method are depicted in Figure \ref{fNetworks}. Because the $\Phi_p$ criteria are information functions, the relative quality of the networks can be meaningfully measured by the ratio of the $\Phi_p$-values, which are shown in Table \ref{tblSmallNetwork}. Despite the simplicity of the problem, the KL exchange algorithm generally managed to improve upon the greedy solution, by up to 27\%.

\begin{figure}[t]
	\centering
	\includegraphics[trim={0cm 4.5cm 0cm 2.0cm}, clip, width = 0.7\columnwidth]{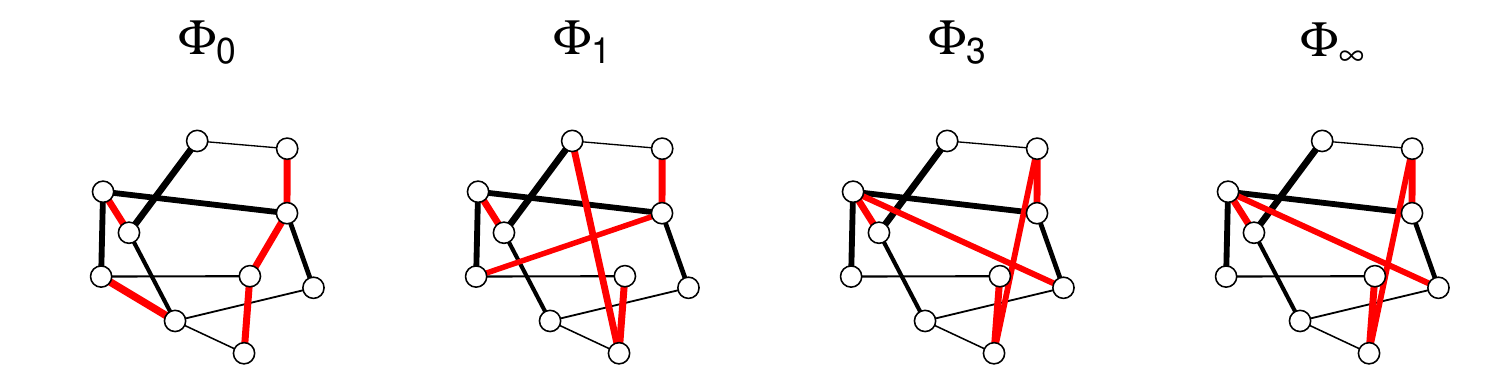} 
	\caption{Networks obtained by the KL exchange algorithm for $D$-, $A$-, $\Phi_{3}$- and $E$-optimality. The black edges denote the original network $E_0$ and the red edges the added $E$. The widths of the edges are proportional to the weights $\w$.}\label{fNetworks}
\end{figure}

\begin{table}[t]
	\centering
    \caption{Results of the greedy ($g$) and KL ($KL$) algorithms for the network in Section \ref{ssSmall}. \label{tblSmallNetwork}}
    \begin{tabular}{c | c c c c}
		\hline
		$p$ & 0 ($D$) & 1 ($A$) & 3 & $\infty$ ($E$) \\ \hline
        $\Phi_p(E_{KL}) / \Phi_p(E_{g})$ &  1.000  & 1.062 & 1.068 & 1.274 \\
        $\Phi_p(E_{KL}) / \Phi_p(E^*)$ &  1.000  & 0.9995 & -- & -- \\
        $\Phi_p(E_{g}) / \Phi_p(E^*)$ &  1.000  & 0.941 & -- & -- \\
        \hline
	\end{tabular}
\end{table}

Table \ref{tblSmallNetwork} also displays, for $D$- and $A$-optimality, the quality of the networks relative to the actually optimal networks (in ED context, this quality relative to the global optimum is called efficiency). The optimal networks were constructed by applying the mixed-integer second-order cone programming formulation \cite{SagnolHarman15} for the regularized ED model with $\widetilde{\x}_e=\K^T\x_e$, using the \texttt{OptimalDesign} package \cite{OptimalDesign} in the R software. The network produced by the greedy method for $D$-optimality is actually $D$-optimal, and the KL algorithm naturally found the same network. For $A$-optimality methods, the obtained networks are not $A$-optimal, but the network obtained by the KL exchange method is nearly optimal, with quality $99.95\%$.

\subsection{Larger networks}
\label{ssLarger}

For each $n=500, 1000, 1500$, we generated edge weights $\w$ from the uniform distribution on $(0,1)$, and constructed an initial network $E_0$ as the union of a random tree on the $n$ nodes and an Erd\H{o}s-R\'enyi graph on the same nodes, with $m=n$ edges. The initial networks therefore were connected and had $n$ nodes, and the number of initial edges $N_0$ was approximately $2n$. The objective was to add $N=500$ edges to $E_0$ in each case. Note that we did not restrict the set of the candidate edges; therefore, the number of candidate edges $k$ was approximately $n(n-5)/2$. These settings are summarized in Table \ref{tblLargeNetworks}. We used the efficient greedy algorithm for $D$-, $A$- and $\Phi_3$-optimality to construct a network, followed by the efficient KL exchange algorithm (with $K=L=20$) to potentially improve the network; see Table \ref{tblLargeResultsD}. We did not consider $E$-optimality here, because we do not have efficient rank-one updates for $E$-optimality and the algorithms therefore tend to become prohibitively slow for such large networks. The results show that Algorithm \ref{aKLeff} is generally able to improve upon the greedy solution in a relatively short time.

\begin{table}[t]
	\centering
    \caption{Characteristics of the networks in Section \ref{ssLarger}. \label{tblLargeNetworks}}
	\begin{tabular}{c | r r r r}
		\hline
		Network & $n$ & $N_0$ & $N$ & $k$ \\ \hline
		1 &  500 & 995 & 500 & 123,755 \\
		2 & 1,000 & 1,998 & 500 & 497,502 \\
		3 & 1,500 & 2,995 & 500 & 1,121,255 \\
		\hline
	\end{tabular}
\end{table}

\begin{table}[t]
	\centering
    \caption{Results of the greedy ($g$) and KL ($KL$) algorithms for the  networks in Section \ref{ssLarger}. \label{tblLargeResultsD}}
	\begin{tabular}{c c | r r c  }
		\hline
		$p$ & Network & $\text{time}_g$ (s) & $\text{time}_{KL}$ (s)  & $\Phi_p(E_{KL}) / \Phi_p(E_{g})$   \\ \hline
		\multirow{3}{*}{0 ($D$)} 
		&1 & 2.99 & 3.56 & 1.0006  \\
		&2 & 23.45 & 11.16 & 1.0001 \\
		&3 & 60.65 & 31.17 & 1.0003  \\
		\hline
		\multirow{3}{*}{1 ($A$)} 
		&1 & 29.63 & 3.75 & 1.0019  \\
		&2 & 56.35 & 4.71 & 1.0004  \\
		&3 & 237.45 & 11.78 & 1.0000  \\
		\hline
		\multirow{3}{*}{3} 
		&1 & 67.50 & 7.98 & 1.0012  \\
		&2 & 174.96 & 17.39 & 1.0003  \\
		&3 & 329.94 & 22.93 & 1.0001 \\
		\hline
	\end{tabular}
\end{table}

\section{Discussion}

We provided theoretical and algorithmic results that can aid in the analysis and construction of well-performing networks, often by utilizing the connection between Network design and Experimental design. The directional derivatives and the node dissimilarities $v_{\Lb, p}$ played a crucial role in the presented results; we therefore highlight them as useful tools for the optimization of Laplacian spectra. 

Some of our results are new, while others extend existing findings. Even in the latter case, those findings are often known only within a single ND area, despite the fact that the Laplacian-optimization problems described in Section \ref{ssAreas} are mathematically equivalent across these areas. Accordingly, results proved in any one ND setting, including ours, prior work, and future advances, apply directly to all equivalent instances and need not be re-derived separately for each domain.

The presented results can be generalized straightforwardly in certain directions. In particular, for conciseness, we have assumed that any edge can be added to the network, but the results extend directly to problems where the graph of potential edges is incomplete. Similarly, the analysis can be easily extended to multigraphs. Non-trivial extensions and other interesting avenues for future research include:
\begin{enumerate}
	\item While we briefly summarized the diverse applications of Laplacian spectral optimization across different areas of ND, a more comprehensive survey would be both valuable and timely.
	\item Criteria other than the $\Phi_p$ measures deserve attention as well. In ED, popular measures include $MV$-, $G$-, and $c$-optimality (\cite{LopezFidalgo2023}, Section 2), which are not eigenvalue-based. Moreover, Kiefer's measures can be extended to $p \in [-1,0)$ (Chapter 6 in \cite{puk}), which, however, have significantly different properties; e.g., they are non-zero even for disconnected networks.
	\item A particularly interesting direction is the optimization of networks with respect to specific systems of comparisons; for example, minimization of the average resistance distance between a predefined subset of nodes $S$ and its complement $V \setminus S$. In ED, such objectives are known as ``partial'' optimality criteria or ``subsystems of interest'' \cite{puk}. 
	\item For $E$-optimality, we used the node dissimilarity $v_{\Lb, \infty}(e)$ for edge selection in the exchange method, which is a common approach in $E$-optimality algorithms. However, we noted in Section \ref{ssDirDer} that when $\lambda_{n-1}(\Lb)$ is not simple, the node dissimilarities and the rank-one directional derivatives $\partial\Phi_{\infty}(\Lb, \Lb + \x_e\x_e^T)$ may not be appropriate to use. Heuristics based on other directional derivatives, superdifferentials or other fundamentally different approaches are therefore worth examining. Alternatively, a thorough theoretical analysis of exchange or greedy heuristics that work with $v_{\Lb, \infty}(e)$ would be beneficial.
	\item We deliberately avoided continuous convex relaxations, which are popular in both ND and ED. However, we plan to investigate the continuous relaxation approach in ND in the future.
	\item Other optimization techniques developed for or used in ED may also be relevant, such as metaheuristics \cite{Haines87, DevonLin}, branch-and-bound methods \cite{Welch, DuarteEA, Ahipasaoglu21}, and the mixed-integer second-order cone programming approach \cite{SagnolHarman15}.
\end{enumerate}

\section*{Acknowledgments}
This work was supported by the Slovak Scientific Grant Agency
(VEGA), Grant 1/0480/26.

{\appendix
\section*{Appendix: Proofs}
\label{apProofs}

\subsection*{Proof of Theorem \ref{tReg}}

\begin{proof}
	Matrix $\K$ satisfies $\K^T\K = \I_{n-1}$ and $\K\K^T = \I_n - \J_n/n$, which is a projection matrix to $\1_n^\perp$. The latter claim follows from the fact that $[\K, \1_n/\sqrt{n}]$ is an orthogonal matrix. It follows that $\K\K^T \Lb = \Lb$, because $\C(\Lb) = \1_n^\perp$.
	
	Observe that $
	\K^T\Lb^+\K \widetilde{\Lb} = \K^T \Lb^+ \Lb \K$. Moreover, assumption $\rank(\Lb)=n-1$ guarantees $\C(\Lb) = \C(\K)$, which implies that there exists a matrix $\Z$, such that $\K = \Lb \Z$. Thus, $\K^T\Lb^+\K \widetilde{\Lb} = \Z^T\Lb \Lb^+ \Lb \K = \Z^T \Lb \K = \K^T\K = \I_{n-1}$, which proves (b) and (a). The positive eigenvalues of $\K^T \Lb \K$ are the same as those of $\Lb \K \K^T$ (see \cite{Seber}, Example 6.54(b)), which is exactly the matrix $\Lb$.
\end{proof}

\subsection*{Proof of Theorem \ref{tDirDer}}

To formulate results that are based on the regularized Laplacians in terms of the original Laplacians, we will employ the following technical lemma.

\begin{lemma}\label{lPowers}
	Let $p \in [0, \infty)$, and let $\Lb$ and $\K$ be as in Theorem \ref{tReg}. Then,
	\begin{equation}
		\label{eKLK}
		(\K^T \Lb \K)^{-p} = \K^T (\Lb^+)^{p} \K.
	\end{equation}
\end{lemma}

\begin{proof}
	Let $p \in [0, \infty)$ and
	let $\Lb = \V \Lambdab \V^T$ be the spectral decomposition of $\Lb$, where $\V=[\vb_1, \ldots, \vb_{n-1}, \1_n/\sqrt{n}]$. Then $\Lb^+ = \V \Lambdab^+ \V^T$, and $\K^T\vb_1, \ldots, \K^T \vb_{n-1}$ are the eigenvectors of $\K^T\Lb\K$ with eigenvalues $\lambda_1, \ldots, \lambda_{n-1}$. Let us denote $\bar{\V} = [\vb_1,\ldots, \vb_{n-1}]$ and $\bar{\Lambdab} = \diag(\lambda_1, \ldots, \lambda_{n-1})$; therefore $\K^T \bar{\V} \bar{\Lambdab} \bar{\V}^T \K$ is the spectral decomposition of $\K^T\Lb\K$. Then,
	\begin{align*}
		(\K^T\Lb\K)^{-p} &= \K^T \bar{\V} \bar{\Lambdab}^{-p} \bar{\V}^T \K 
		=\K^T \V \begin{bmatrix}
			\bar{\Lambdab}^{-p} & 0 \\ 0 & 0
		\end{bmatrix} \V^T \K \\
		&= \K^T \V (\Lambdab^+)^{p} \V^T \K 
		= \K^T (\Lb^+)^{p} \K.
	\end{align*}
\end{proof}

Now, we can prove Theorem \ref{tDirDer}.

\begin{proof}
	Recall that $\K\K^T = \I_{n} - \J_{n}/n$ is the projection matrix to $\1_n^\perp$, which yields 
	\begin{equation}\label{eLProj}
		\K\K^T\Lb_i = \Lb_i\K\K^T = \Lb_i, \quad i=1,2. 
	\end{equation}
	Moreover, $\Phi_p(\Lb_i) = \widetilde{\Phi}_p(\K^T\Lb_i\K)$ (Theorem \ref{tReg}).
	The formula for the directional derivative of $\widetilde{\Phi}_p$ is
	\begin{equation}\label{ePartialPhip}
		\partial\widetilde{\Phi}_p(\A, \B) = \widetilde{\Phi}_p(\A)\left(\frac{\tr(\A^{-p-1}\B)}{\tr(\A^{-p})} - 1 \right)
	\end{equation}
	\cite{Gaffke85}. It follows that
	\begin{align*}
		\partial{\Phi}_p(\Lb_1, \Lb_2) 
		&= \partial\widetilde{\Phi}_p(\K^T\Lb_1\K, \K^T\Lb_2\K) \\
		&= \widetilde{\Phi}_p(\K^T\Lb_1\K)\left(\frac{\tr((\K^T\Lb_1\K)^{-p-1}\K^T\Lb_2\K)}{\tr((\K^T\Lb_1\K)^{-p})} - 1 \right) \\
		&=\Phi_p(\Lb_1)\left(\frac{\tr(\K^T(\Lb_1^+)^{1+p}\K\K^T\Lb_2\K)}{\tr(\K^T(\Lb_1^+)^{p}\K)} - 1 \right) \quad [\eqref{eKLK}] \\
		&=\Phi_p(\Lb_1)\left(\frac{\tr((\Lb_1^+)^{1+p}\Lb_2)}{\tr((\Lb_1^+)^{p})} - 1 \right) \quad [\eqref{eLProj}].
	\end{align*} 
\end{proof}

\subsection*{Proof of Theorem \ref{tDirDerE}}

Before tackling the proof, we first extend the function $\Phi_\infty$ to the entire set $\Sc^{n}$: $\Psi(\A) = \lambda_{n-1}(\A) + \lambda_n(\A)$ for $\A \in \Sc^{n}$. Clearly, $\Phi_\infty(\Lb) = \Psi(\Lb)$ for any $\Lb \in \Lc^n$. It follows that
$\partial\Phi_\infty(\Lb_1, \Lb_2) = \partial\Psi(\Lb_1, \Lb_2)$
for any $\Lb_1, \Lb_2 \in \Lc^n$. Note that $\Psi$ is non-differentiable even at some matrices from $\Lc^n$ of rank $n-1$.

The superdifferential of a function $f: \Sc^n \to \R$ at $\A$ is the set
\begin{equation*}
S_{f}(\A) = \{ \G \in \Sc^n \, \vert \, f(\B) \leq f(\A) + \tr(\G(\B - \A)) \,\, \forall \, \B \in \Sc^n \}.
\end{equation*}
Any element of $S_{f}(\A)$ is called a supergradient of $f$ at $\A$. When $f$ is concave, but not necessarily differentiable at $\A$, then 
\begin{equation}\label{eSuperDirDer}
	\partial f(\A, \B) = \inf_{\G \in S_f(\A)} \tr( \G (\B - \A))
\end{equation}
(cf. Theorem 23.4 in \cite{Rockafellar}).

\begin{lemma}\label{lSuperE}
	Let $\Lb$ be the Laplacian of a connected graph and let $s$ be the multiplicity of $\lambda_{n-1}(\Lb)$. Then,
	\begin{equation*}
	S_{\Psi}(\Lb) = \left\{\left. \frac{1}{n}\J_n + \sum_{i=1}^s \gamma_i \ub_i \ub_i^T \, \right\vert \, [\ub_1, \ldots, \ub_s] \in \Uc, \gammab \in \Gamma \right\},
\end{equation*}
	where $\Uc$ is the set of $(n-1)\times s$ matrices, whose columns are orthonormal eigenvectors of $\Lb$ corresponding to $\lambda_{n-1}(\Lb)$, and $\Gamma = \{\gammab \in \R^s \, \vert \, \sum_{i=1}^s \gamma_i = 1, \gamma_i \in [0, 1] \, \forall i \}$.
\end{lemma}

\begin{proof}
	The result follows directly from Proposition 3 in \cite{Harman04}, by setting $k=2$ in their notation.
\end{proof}

Proof of Theorem \ref{tDirDerE} follows.
\begin{proof}
	\begin{align*}
		\partial\Phi_{\infty}(\Lb_1, \Lb_2) 
		&= \partial\Psi(\Lb_1, \Lb_2)
		= \inf_{\G \in S_\Psi(\Lb_1)} \tr( \G (\Lb_2 - \Lb_1)) \\
		&= \inf_{\gammab \in \Gamma, (\ub_1, \ldots, \ub_s) \in \Uc} \sum_{i=1}^s \gamma_i \ub_i^T\Lb_2\ub_i - \lambda_{n-1}(\Lb_1)
	\end{align*}
	because of Lemma \ref{lSuperE}, $\J_n\Lb_i = \0$, and $\tr(\ub_i\ub_i^T\Lb_1) = \lambda_{n-1}(\Lb_1)$.
	For a fixed $[\ub_1, \ldots, \ub_s] \in \Uc$, we have \begin{equation*}\inf_{\gammab \in \Gamma} \sum_{i=1}^s \gamma_i \ub_i^T\Lb_2\ub_i = \min_{i=1, \ldots, s} \ub_i^T\Lb_2\ub_i.\end{equation*} Moreover, 
	\begin{equation}\label{ePartialEmin}
		\inf_{(\ub_1, \ldots, \ub_s) \in \Uc} \min_{i=1, \ldots, s} \ub_i^T\Lb_2\ub_i = \inf_{\ub} \ub^T\Lb_2\ub,
	\end{equation} 
	where the infimum on the right-hand side is taken over all unit-length eigenvectors $\ub$ of $\Lb_1$ corresponding to $\lambda_{n-1}(\Lb_1)$. 
	We obtained the infimum of a continuous function over a compact set; therefore, it is actually the (finite) minimum. 
	To solve the optimization problem in \eqref{ePartialEmin}, note that $\Lb_1 \ub = \lambda_{n-1}(\Lb_1)\ub$ is equivalent to $\ub = \U\z$ for some $\z \in \R^s$. Then, $1 = \lVert \ub \rVert^2 = \lVert \z \rVert^2$, because the columns of $\U$ are orthonormal, and $\ub^T\Lb_2\ub = \z^T\U^T\Lb_2\U\z$. Moreover,
	\begin{equation*}
	\min_{\lVert \z \rVert^2 = 1} \z^T\U^T\Lb_2\U\z = \lambda_{s}(\U^T\Lb_2\U).
\end{equation*}
	If $\lambda_{n-1}(\Lb_1)$ is simple, then $\U = \ub$ and $\lambda_{s}(\U^T\Lb_2\U) = \ub^T\Lb_2\ub$.
	
	Equality \eqref{ePartialDerqE} is obtained by setting $\Lb_2=\Lb_1 + \x_e\x_e^T$. If $\lambda_{n-1}(\Lb_1)$ is simple, then $\ub^T\Lb_2\ub = (\x_e^T\ub)^2 + \lambda_{n-1}(\Lb_1)$. If it is not simple, then 
	$
	\lambda_{s}(\U^T\Lb_2\U) = \lambda_{s}(\U^T\x_e\x_e^T\U) + \lambda_{n-1}(\Lb_1).
	$
	Then, $\lambda_s(\U^T\x_e\x_e^T\U) = 0$, because the rank of $\U^T\x_e\x_e^T\U \in \Sc^s_+$ is at most 1.
\end{proof}

\subsection*{Proof of Theorem \ref{tPowerTrace}}

We first formulate a technical lemma.

\begin{lemma}\label{lPowerMat}
	Let $\A \in \R^{m \times m}$ be non-singular, and let $\ub, \vb \in \R^n$. Then,
	\begin{equation}\label{ePowerGenralWOtrace}
		(\A + \ub\vb^T)^p = \A^p + \sum_{k=1}^p \sum_{\substack{i_0, \ldots, i_k \geq 0 \\ i_0 + \ldots + i_k = p-k}} \A^{i_0} \ub \left( \prod_{j=1}^{k-1} \vb^T \A^{i_j} \ub \right) \vb^T \A^{i_k}
	\end{equation}
	and
	\begin{equation}\label{ePowerGeneral}
		\tr((\A + \ub\vb^T)^p) = \tr(\A^p) + \sum_{k=1}^p \sum_{\substack{i_0, \ldots, i_k \geq 0 \\ i_0 + \ldots + i_k = p-k}} \vb^T \A^{i_0 + i_k} \ub \left( \prod_{j=1}^{k-1} \vb^T \A^{i_j} \ub \right).
	\end{equation}
	for any positive integer $p$. Here, $i_0, \ldots, i_k$ are integers, and the product over $j \in \{1, \ldots, k-1\}$ is understood to be equal to one if $k=1$.
\end{lemma}

\begin{proof}
	The claim can be proved by mathematical induction. The formula clearly holds for $p=1$. Now suppose that it holds for $p$; we wish to prove that
	\begin{equation}
		(\A + \ub\vb^T)^{p+1} 
		= \A^{p+1} + 
		\sum_{k=1}^{p+1} \sum_{\substack{i_0, \ldots, i_k \geq 0 \\ i_0 + \ldots + i_k = p-k+1}} \A^{i_0} \ub \left( \prod_{j=1}^{k-1} \vb^T \A^{i_j} \ub \right) \vb^T \A^{i_k}. \label{ePowerFormula}
	\end{equation}
	We have
	\begin{align}
		(\A + \ub\vb^T)^{p+1} 
		&= (\A + \ub\vb^T) (\A + \ub\vb^T)^p \notag \\
		&= \A^{p+1} + 
		\sum_{k=1}^p \sum_{\substack{i_0, \ldots, i_k \geq 0 \\ i_0 + \ldots + i_k = p-k}} \A^{i_0 + 1} \ub \left( \prod_{j=1}^{k-1} \vb^T \A^{i_j} \ub \right) \vb^T \A^{i_k} + 
		\ub\vb^T\A^p +\notag  \\
		&+\sum_{k=1}^p \sum_{\substack{i_0, \ldots, i_k \geq 0 \\ i_0 + \ldots + i_k = p-k}} \ub\vb^T\A^{i_0} \ub \left( \prod_{j=1}^{k-1} \vb^T \A^{i_j} \ub \right) \vb^T \A^{i_k}  \notag \\
		&= \A^{p+1} + 
		\sum_{z=1}^p \sum_{\substack{\ell_0 \geq 1, \ell_1 \ldots, \ell_z \geq 0 \\ \ell_0 + \ldots + \ell_z = p-z+1}} \A^{\ell_0} \ub \left( \prod_{j=1}^{z-1} \vb^T \A^{\ell_j} \ub \right) \vb^T \A^{\ell_z} +
		\ub\vb^T\A^p + \notag \\
		&+\sum_{z=2}^{p+1} \sum_{\substack{\ell_1, \ldots, \ell_z \geq 0 \\ \ell_1 + \ldots + \ell_z = p-z+1}} \ub \left( \prod_{j=1}^{z-1} \vb^T \A^{\ell_j} \ub \right) \vb^T \A^{\ell_z}, \label{ePowerFormula2}
	\end{align}
	where the last equality is obtained by appropriate substitutions. Each element in the main sum in \eqref{ePowerFormula} (where $k$ runs from $1$ to $p+1$) can be split into the part with $i_0 = 0$ and into that with $i_0 \geq 1$. The second term in \eqref{ePowerFormula2} corresponds to the main sum in \eqref{ePowerFormula} for $i_0 \geq 1$; the third term corresponds to the main sum in \eqref{ePowerFormula} for $i_0 = 0$ and $k=1$; and the last term in \eqref{ePowerFormula2} corresponds to the main sum in \eqref{ePowerFormula} for $i_0 = 0$ and $k \geq 2$; thus \eqref{ePowerFormula2} is equal to \eqref{ePowerFormula}. \eqref{ePowerGeneral} is a direct corollary of \eqref{ePowerGenralWOtrace}
\end{proof}

Theorem \ref{tPowerTrace} is obtained by combining \eqref{ePowerGeneral} and \eqref{eRankOneUpdate}.

\subsection*{Proof of Theorem \ref{tPowerInv}}

We again begin with a technical lemma.

\begin{lemma}\label{lPowerInv}
	Let $\A \in \R^{n \times n}$ be non-singular, and let $\ub, \vb \in \R^n$. Then,
	\begin{equation*}
	(\A+\ub\vb^T)^p = \A^p + \sum_{j=0}^{p-1} (\A+\ub\vb^T)^j\ub\vb^T\A^{p-j-1}
\end{equation*}
	for any positive integer $p$.
\end{lemma}

\begin{proof}
	The proof can be approached using mathematical induction. For $p=1$, the claim is clearly valid. Suppose that the formula holds for $p$; then
	\begin{align*}
		(\A + \ub\vb^T)^{p+1} 
		&= (\A + \ub\vb^T)^{p}(\A + \ub\vb^T)
		= \A^{p+1} + \sum_{j=0}^{p-1} (\A+\ub\vb^T)^j\ub\vb^T\A^{p-j} + (\A+\ub\vb^T)^p\ub\vb^T \\
		&= \A^{p+1} + \sum_{j=0}^{p} (\A+\ub\vb^T)^j\ub\vb^T\A^{p-j},
	\end{align*}
	which completes the proof.
\end{proof}

Theorem \ref{tPowerInv} follows by combining Lemma \ref{lPowerInv} with \eqref{eRankOneUpdate}.
}

\bibliographystyle{IEEEtran}
\bibliography{rosa.bib}

\end{document}